\documentclass[review]{elsarticle}

\usepackage{lineno,hyperref}
\modulolinenumbers[1]

\biboptions{authoryear}

\usepackage{natbib}
\usepackage[english]{babel}
\usepackage{amsthm,amsmath}
\usepackage{graphicx}
\usepackage{epstopdf}
\usepackage{amsfonts,amssymb}
\usepackage{epsfig,color}

\usepackage{algorithmicx}
\usepackage{algpseudocode}
\usepackage{algorithm}

\newcommand{\lo}{\mathrm{low}}
\newcommand{\up}{\mathrm{upp}} 
\newcommand{\pmin}{\mathrm{min}}
\providecommand{\1}{\mathbf{1}} 

\newtheorem{theorem}{Theorem}[section]

\begin{document}

\begin{frontmatter}

\title{New methods for multiple testing in permutation inference for the general linear model}

\author{Tom\'a\v s Mrkvi\v cka$^{1,2}$}
\address{Faculty of Economics, University of South Bohemia, Czech Republic}
\fntext[myfootnote]{Corresponding author: mrkvicka.toma@gmail.com}
\fntext[myfootnote]{Dpt. of Applied Mathematics and Informatics, Faculty of Economics, University of South Bohemia, Studentsk{\'a} 13, 37005 \v{C}esk\'e Bud\v{e}jovice, Czech Republic}

\author{Mari Myllym{\"a}ki}
\address{Natural Resources Institute Finland (Luke), Helsinki, Finland}

\author{Mikko Kuronen}
\address{Natural Resources Institute Finland (Luke), Helsinki, Finland}

\author{Naveen Naidu Narisetty}
\address{University of Illinois, Urbana-Champaign, USA}




\begin{abstract}
Permutation methods are commonly used to test significance of regressors of interest in general linear models (GLMs) for functional (image) data sets, in particular for neuroimaging applications as they rely on mild assumptions.
Permutation inference for GLMs typically consists of three parts:
choosing a relevant test statistic,
computing pointwise permutation tests and
applying a multiple testing correction.
We propose new multiple testing methods as an alternative to the commonly used maximum value of test statistics across the image. The new methods improve power and robustness against inhomogeneity of the test statistic across its domain. The methods rely on sorting the permuted functional test statistics based on pointwise rank measures; still they can be implemented even for large brain data. The performance of the methods is demonstrated through a designed simulation experiment, and an example of brain imaging data.
We developed the R package GET which can be used for computation of the proposed procedures.

\end{abstract}


\end{frontmatter}

%\linenumbers

\section{Introduction}

General linear models (GLMs) are among the most commonly used statistical tools in the field of neuroimaging for analyzing imaging data \citep{Christensen2002}. To perform statistical inference based on GLMs, one approach is to use parametric methods which rely on stringent assumptions, such as the normality of the random errors in the GLM, homogeneity and independence of the random errors across the image. Alternatively, one can use non-parametric methods with weaker assumptions about the data \citep[see e.g.][Table 1]{WinklerEtal2014}. 

Permutation tests are a class of non-parametric methods, which have a long history going back to \citet{Fisher1935}. 
Fisher demonstrated that the null hypothesis could be tested by observing, 
how often the test statistic computed from permuted observations would be more extreme than the same statistic computed without permutation.
Even though the data in neuroimaging applications are commonly images, the same principle still holds;
a detailed introduction to the permutation inference for the GLMs for neuroimage data can be found in \citet{WinklerEtal2014}.
Another advantage of the non-parametric methods is that they do not suffer from the false positive rates reported by \citet{EklundEtal2016} for parametric approaches. 

Analysis based on permutation tests consists of several crucial steps. First, a suitable test statistic that is computed for each single {location (}spatial point, voxel, vertex or face) of an image has to be chosen such that it is informative about the studied null and alternative hypotheses and homogeneous across the image. Second, the appropriate  permutation scheme has to be used to generate the permutations from the studied null hypothesis. The permutations are performed between subjects, thus throughout the whole paper we discuss the second-level analysis. Finally, an appropriate multiple testing correction has to be applied to significance results obtained for all spatial {locations} by the permutation test. 

Commonly used test statistics include the $t$- and $F$-statistics,
or the $G$ statistic defined in \citet{WinklerEtal2014} which generalises the classical statistics into various cases with heteroscedasticity.
However, a serious limitation of the $F$ or $G$ statistics is that they are not pivotal across different locations of the image in terms of the entire distribution, but only in terms of the first and second moments when the errors are not normally distributed. Therefore, if the error distribution is non-normal and inhomogeneous across the image, the distribution of these test statistics varies across the image causing the desired quantiles of the test statistics to vary as well. For example, if the error term has lognormal distribution with parameters 1 and 2 on one side of the image and with parameters 1 and 0.5 on another part of the image, then the ratio of 95\% quantiles of the $F$-statistic computed from 20 observations is equal to 1.37. The ratio of 99\% quantiles is 1.64. (These ratios were computed from simulated sample of 10000 $F$-statistics under the null hypothesis.) If this heterogeneity is ignored, e.g.\ by taking the maximum $F$-statistic across the spatial points as the test statistic, 
some or all spatial points where the null model is broken may be overlooked. 
In quantitative manner, this can bring a substantial loss of power as we demonstrate in our simulation study. 

Due to the non-pivotal nature of the $F$ and $G$ statistics, a more suitable choice is the permutation $p$-value, which is automatically a pivotal statistic, i.e.\ its distribution does not change across the image. Therefore, it would be convenient to perform the multiple testing correction by taking the minimum $p$-value across the image as the test statistic. However, the permutation $p$-value has another serious disadvantage: due to its discreteness, 
the (minimum) $p$-values obtained from permutations contain ties and the resulting test tends to be conservative, leading to loss of power as well \citep{PantazisEtal2005}.
This disadvantage is enormous especially in the case of high resolution images where the number of permutations can not be large due to computational limitations. This is further demonstrated in the introductory example below.

Another problem arises for test statistics which accumulate the data from the surrounding area of a single spatial {location}, e.g.\ the local area above a given threshold around a spatial {location}. Namely, if the spatial autocorrelation of the data is inhomogeneous across the image, then the sizes of the areas above the thresholds in different parts of the image are not comparable. This problem was recently treated for the special case of cluster size permutation tests in \citet{HAYASAKA2004676} and \citet{SALIMIKHORSHIDI20112006}.
In this paper we discuss solutions for the case of inhomogeneous spatial autocorrelation and also for the case of inhomogeneous distribution of the test statistic.

The classical multiple testing procedures, computing the maximum $F$ ($F$-max) or the minimum $p$-value ($p$-min) statistic across all spatial points, control the family wise error rate (FWER) \citep[see e.g.][]{WinklerEtal2014}. 
This paper introduces three new multiple testing corrections for the permutation inference for the GLMs, controlling the FWER in a similar manner. 

The new methods are alternative solutions to the ties problem of $p$-values.
The first correction is based on the extreme rank length (ERL) measure that was originally proposed by \citet{MyllymakiEtal2017} and \citet{NarisettiNair2016} for spatial and functional data analysis.
The other two corrections rely on the continuous rank  measure (Cont) introduced in the technical report of \citet{Hahn2015} and the area rank measure (Area) presented in this work first time. 
We carefully describe the homogeneity assumptions of the $F$-max, $p$-min and new methods and further their sensitivity to different types of extremeness of the test statistic.
By a simulation study, 
we illustrate the power of the new and existing tests under different scenarios, both when the homogeneity assumptions of the multiple correction methods are met and when not.
A further asset of the methods, namely a graphical interpretation of the test results in a similar manner as in global envelope testing \citep{MyllymakiEtal2017, MrkvickaEtal2016, MrkvickaEtal2018}, is illustrated by a real data example of autism brain imaging data \citep{MartinoEtal2014}.
We conclude that the choice of the measure for multiple correction should be done on the basis of reasonable assumptions of different types of homogeneity, and the expected type of extremeness of the test statistic.

{The proposed tests are implemented in the R library GET \citep{MyllymakiEtal2017, MyllymakiMrkvicka2020}.}

\subsection{Introductory example}\label{sec:introexample}

To illustrate the proposed methods, we studied the autism brain imaging data collected by resting state functional magnetic resonance imaging (R-fMRI) \citep{MartinoEtal2014}. The raw fMRI data contains measurements from 539 individuals with the autism spectrum disorder (ASD) and 573 typical controls (TC). When {the number of individuals in the analysis is large}, the test statistics can be well approximated by the limiting $F$ distribution. Therefore, in order to show the differences between the methods, we have chosen a small number of individuals{, namely the patients measured in Oregon Health and Science University with 13 subjects with Autism Spectrum Disorders (ASD), 15 Typical Controls (TC) and homogeneity in sex}. In fact, when the whole dataset of 1112 patients was used in our analysis, 
the $F$-max procedure worked equally to our proposed methods. 
The imaging measurement for local brain activity at resting state was fractional amplitude of low frequency fluctuations \citep{ZouEtal2008}.
The whole brain was analyzed consisting of 175 493 voxels. The difference between the groups was studied while age was taken as a nuisance factor. The $F$-max, $p$-min and new tests were performed using the same 50 000 and 100 000 permutations. The experiment was run using the R package GET \citep{MyllymakiMrkvicka2020} and took 10 hours using two cores and 10 Gb memory of an ordinary laptop. Most of this time was allocated to calculating the $F$-statistics for each permutation and each voxel, which were needed for all tests.

Table \ref{tab:pval} shows the $p$-values of all the investigated methods. 
While the differences in the $p$-values are small between the $F$-max and new methods, only the ERL and Area tests rejected 
the null hypothesis of no group effect at the strict significance level 0.05.
An essential difference between the $F$-max and new methods is in their critical bounds:
Figure \ref{fig:upenv} shows the 95\% global upper envelope of the global Area test for one slice of the brain; Appendix B presents the results for the whole brain. 
This upper envelope is adjusted to the variability of the 95\% global quantile of the $F$-statistics, ranging from 20 to 70 for the given slice. If the empirical $F$-statistic crosses the 95\% global quantile, the null hypothesis is rejected. 
On the other hand, the critical bound of the $F$-max test is constant (48.4), therefore the $F$-max procedure can overlook significant voxels where the variability of the $F$-statistics is smaller.

\begin{figure}[ht!]
    \centering
    \includegraphics{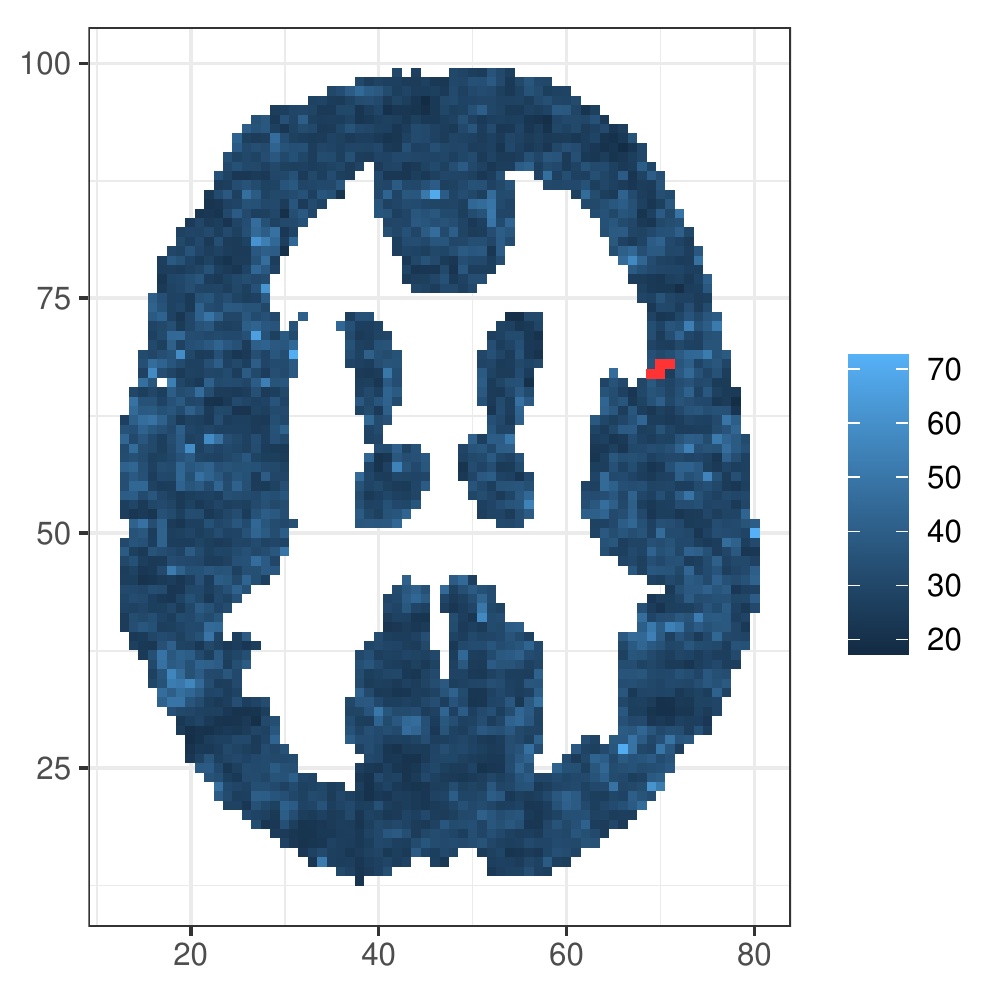}
    \caption{Upper envelope (blue) of the global area rank test and significant voxels (red) on slice 45 when testing the effect of the group variable in abide brain data subset using 100\,000 permutations. }
    \label{fig:upenv}
\end{figure}

\begin{table}[ht]
\centering 
\begin{tabular}{rrrrrr}
 \hline
Number of permutations & $F$-max & $p$-min & ERL & Cont  & Area \\
\hline
50000 &  0.058 &  0.592 &  0.032 &  0.105 &  0.040\\
100000 &  0.058 &  0.420 &  0.041 &  0.053 &  0.040\\
\hline
\end{tabular}
\caption{The $p$-values for the test of the effect of the group variable in the example brain data subset computed by investigated methods. The ERL, Cont and Area refer to the new methods to be introduced in this paper.}\label{tab:pval}
\end{table}

The $p$-min procedure, which is also adjusted to the variability of the statistics across the image, was highly conservative for 50 000 or 100 000 permutations (Table \ref{tab:pval}). 
Figure \ref{fig:rank1ties} further shows the proportion of permutations which achieved the most extreme $F$-statistic at least for one voxel for different numbers of permutations. This value corresponds to the minimal $p$-value that the $p$-min procedure can achieve. This example with the total of 10 000 000 permutations was run on a computing server with hundreds of cores. Thus if one needs the $p$-min procedure to give an answer for the test of the whole brain, 6 000 000 of permutations has to be used in minimum to achieve the minimal $p$-value 0.02 (which is still not very precise procedure){, and this is typically too much}.

The aim of this paper is to give three solutions for breaking the ties of the $p$-min procedure which can be used with reasonable number of permutations, i.e.\ in the time comparable to the time required by the $F$-max procedure. The proposed methods also provide the correct significance level as it is shown by simulation study in Section \ref{se:se} and also by an experiment similar to one proposed in \citet{EklundEtal2016}, which we report in Appendix A.

\begin{figure}[ht!]
    \centering
    \includegraphics[width=\textwidth]{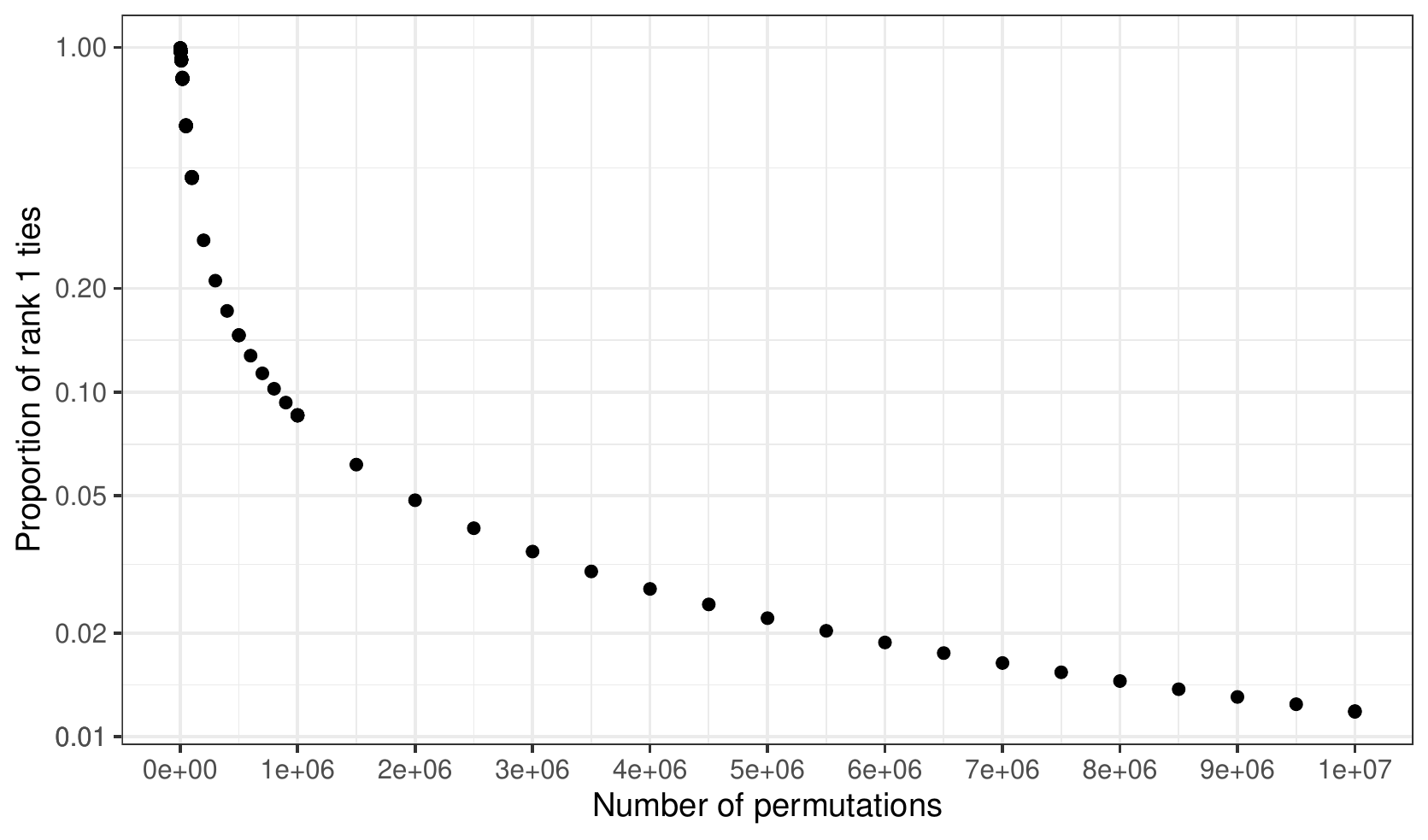}
    \caption{Proportion of rank 1 ties in the $p$-min procedure with respect to the number of permutations, when testing the effect of the group variable in the example brain data.}
    \label{fig:rank1ties}
\end{figure}

\section{Multiple testing correction for permutation methods}

Let us consider the GLM
\begin{equation}\label{eq:GLM}
\mathbf Y(r) = \mathbf X(r) \beta (r) + \mathbf Z(r) \gamma (r) + \epsilon (r), 
\end{equation}
where the argument $r \in I \subset \mathbb R^d$ determines a
spatial point of the image or, generally, {a location where the $d$-dimensional functions, $d\geq 1$, are observed}. The number of {locations} $r \in I$ will be denoted by $n$. 
For every {location} $r$, we consider a one-dimensional GLM with $\mathbf X(r)$ being a $n \times k$ matrix of regressors of interest, $\mathbf Z(r)$ being a $n \times l$ matrix of nuisance regressors, $\mathbf Y(r)$ being a $n \times 1$ vector of observed data and $\epsilon (r)$ being a $n \times 1$ vector of random errors with mean zero and finite variance $\sigma^2(r)$ for every $r \in I$. Further assumptions about the error structure will be given along the definitions of the multiple correction methods below. 
Further, $\beta(r)$ and $\gamma(r)$ are the regression coefficient vectors of dimensions $k \times 1$ and $l \times 1$, respectively.
Often factors are given for the whole image, as in our examples as well. In this case the regressors do not depend on the index $r$ and so the matrices $\mathbf X(r)$ and $\mathbf Z(r)$ are identical for every $r \in I$. However, this simplification is not necessary.
The null hypothesis to be tested is 
$$H_0: \mathbf C (r) \beta (r) = 0,\quad \forall\, r \in I, $$ 
where $\mathbf C (r)$ is a $t \times k$ matrix of $t$ contrasts of interest. 

The GLM setup includes a wide range of applications in neuroimaging, see \citet{WinklerEtal2014} for worked examples. Under various setups, a nonparametric pointwise permutation test can be applied to obtain a set of $p$-values, {$\{p(r): r \in I\}$}. 
We refer to the detailed description of pointwise permutation tests in \citet{WinklerEtal2014}.
Shortly, first a test statistic $T(r)$ is chosen from a vast number of possibilities, reflecting the tested null model $H_0$ and heteroscedasticity. A common choice is the $F$-statistic \citep{Christensen2002},
but in principle any statistic where extreme values reflect evidence against the null hypothesis could be used. 
Then permutations are used to obtain the distribution of $T(r)$ under the null hypothesis $H_0$. The choice of the permutation scheme is important for the performance of the method. We chose the permutation of the residuals under the null model \citep{FreedmanLane1983}, which is approximate, but according to \citet{AndersonBraak2003} it is the most precise permutation method in the case of nuisance effects \citep[see also][]{WinklerEtal2014}. If there are no nuisance regressors, except the constant, permutation of raw data can be used instead; it is however equivalent to performing the Freedman-Lane algorithm since the fit of the null model is constant in this case. In this case, the test is exact.
The last step is to apply a multiple testing correction.

Following sections describe various multiple testing corrections in these permutation procedures.
We assume that the same $J$ permutations have been applied for every location $r\in I$.
We denote the test statistic calculated for the original data by $T_0(r)$ and the corresponding test statistics calculated for each permutation $j$, $j=1,\dots,J$, by $T_1(r), \dots, T_J(r)$.
{Further, $p_{0}(r)$, $r\in I$, denotes the set of pointwise $p$-values for the test statistic $T_0(r)$ of the original data, and
by $p_{j}(r)$ the corresponding $p$-values for each permutation $j$, $j = 1, \ldots , J$.}

\subsection{The $p$-min approach}

In the $p$-min approach, the minimum of the pointwise $p$-values,
\begin{equation}
p^\text{min}_j = \min_{r \in I} p_{j}(r),     
\end{equation}
are calculated for all $j = 0, 1, \dots, J$, and 
the Monte Carlo $p$-value is obtained as
$$ 
p^\text{min} = \frac{1}{J+1} \sum_{j=0}^J \mathbb I (p^\text{min}_j \leq p^\text{min}_0).
$$ 
Here $\mathbb I$ denotes the indicator function. 

The pointwise {$p_j(r)$}-values are pivotal, therefore the $p$-min approach needs no further assumptions about the random error $\epsilon (r)$.
However, since the $p$-values of the permutation test can achieve only values $\frac{1}{J+1}, \frac{2}{J+1}, \ldots , 1$, many ties can appear {among $p^\text{min}_0, p^\text{min}_1, \dots, p^\text{min}_J$}, especially for large $n$, 
leading to a conservative test.

\subsection{Refinements of the $p$-min approach} 

The $p$-min approach is in fact equivalent to the global rank envelope test defined by \citet{MyllymakiEtal2017} for spatial processes. I.e., 
the $p_j^\text{min}$ is after a {simple} normalization equivalent to the extreme rank measure $R_j$: 
\begin{equation}\label{eq:extremerank}
(J+1) p_j^\text{min} = R_j = \min_{r\in I} (J+2-R_j(r)),    
\end{equation} % R_j = (J+1)\min_r p_{rj}
where $R_j(r)$ is the pointwise rank of the test statistic $T_j(r)$ among $\{T_0(r), \ldots , T_J(r)\}$ such that the smallest value obtains the rank 1. The measures and the corresponding global envelopes are defined here as one-sided since the extremeness of the test statistic $T(r)$ used in the GLM is usually realized only for high values. This one-sided alternative also leads to {$J+2$} in the right hand side of the formula \eqref{eq:extremerank}.

In the sequel, we propose the following three methods for breaking the ties between the $p_j^\text{min}$ values:
\begin{itemize}
    \item The extreme rank length (ERL) approach takes into account the number of {pointwise $p$-values equalling the minimal $p$-value $p_j^{\min}$}, 
    \item the continuous rank (Cont) approach measures the {maximal} size of extremeness of $T_j(r)$ {that is associated with $p_j^\text{min}$}, and 
    \item the area rank (Area) {accumulates the extremeness of $T_j(r)$ across the spatial points $r$ where $p_j^{\min}$ is obtained}.
\end{itemize}

\subsubsection{Extreme rank length (ERL)}

The ERL measure refines the $p$-min approach in the sense that not only the minimal $p$-value $p^\text{min}_j$, but also the size of the domain where $p_j(r)$ are equal to $p^\text{min}_j$ is taken into account.
Typically, based on our experience, the ties are already broken by considering the counts of the smallest pointwise $p$-values among $r\in I$. However, to define the ERL measure \citep{MyllymakiEtal2017, NarisettiNair2016, MrkvickaEtal2018}  formally, the complete vectors of pointwise ordered $p$-values
$\mathbf{p}_j=(p_{j[1]}, p_{j[2]}, \dots , p_{j[n]})$, where
the pointwise $p$-values $p_j(r)$ observed at $r_1, \dots, r_n$ are arranged from smallest to largest such that 
$p_{j[k]} \leq p_{j[k^\prime]}$ whenever $k \leq k^\prime$, are considered and ordered by reverse lexical ordering:
The ERL measure is equal to
\begin{equation}
\label{eq:lexicalrank}
   e_j = \frac{1}{J+1}\sum_{j^\prime=0}^{{J}} {\mathbb I}(\mathbf{p}_{j^\prime} \prec \mathbf{p}_j),
\end{equation}
where 
\[
\mathbf{p}_{j^\prime} \prec \mathbf{p}_{j} \quad \Longleftrightarrow\quad
  \exists\, l\leq n: p_{{j^\prime}[k]} = p_{j[k]} \,\forall\, k < l,\  p_{{j^\prime}[l]} < p_{j[l]}
\]
and $J+1$ scales the values to the interval from 0 to 1.
Thus, the ERL measure uses in fact all the pointwise $p$-values to order the test statistics $T_j(r)$ from most extreme to the least extreme one, while the $p$-min procedure utilizes only $p_j^{\min} = p_{j[1]}$.

Since the probability of having a tie in the ERL measure is rather small, the application of Monte Carlo test on the ERL solves the ties problem \citep[see also][]{MyllymakiEtal2017}.
Thus the final $p$-value is $p^\text{erl}= \frac{1}{J+1} \sum_{j=0}^J \mathbb I (e_j \leq e_0)$.

Because often even the counts of the most extreme $p$-values break the ties, the measure can be efficiently implemented utilizing the counts of a certain number (e.g.\ six) of smallest pointwise $p$-values (see Section \ref{sec:computationalissues}).

\subsubsection{Continuous rank (Cont)}

Another refinement of $ p^\text{min}$ is the continuous rank measure 
\begin{equation}\label{eq:cont}
c_j = \frac{1}{J+1} \min_{r} C_j(r) \quad\text{with}\quad C_{j}(r) = J+1-c_{j}(r),
\end{equation}
where $c_{j}(r)$ is the continuous pointwise rank that is a refinement of the ordinary rank $R_j(r)$ based on the relative magnitude of the test statistic $T_{j}(r)$ with respect to the other $T_{j'}(r)$, $j'=0,\dots,J$. The continuous rank $c_{j}(r)$ is defined such that the smallest value of $T_0(r),\dots,T_{J}(r)$ obtains the smallest continuous rank:
Let $T_{[0]}(r) \le T_{[1]}(r) \le \dots \le T_{[J]}(r)$ denote the ordered set of values $T_{j}(r), i=0,1, \dots, J$. Then {the continuous rank of $T_{[j]}(r)$ is}
\begin{equation*}
 c_{[j]}(r) = j + \frac{T_{[j]}(r)-T_{[j-1]}(r)}{T_{[j+1]}(r)-T_{[j-1]}(r)}, \quad\text{for}\; j = 1, 2, \dots, J-1, 
\end{equation*}
and
\begin{align}\label{eq:extremecontranks}
 c_{[0]}(r) = \exp\left(-\frac{T_{[1]}(r)-T_{[0]}(r)}{T_{[J]}(r)-T_{[1]}(r)}\right),\, 
 c_{[J]}(r) =  J - \exp\left(-\frac{T_{[J]}(r)-T_{[J-1]}(r)}{T_{[J-1]}(r)-T_{[0]}(r)}\right).
\end{align}
If the probability to have ties among $T_{j}(r), j=0, \ldots J$, is zero, then the probability of ties among $c_j(r)$ is zero as well. If ties appear among $T_{j}(r), j=0, \ldots J$, i.e. $T_{[j-1]}(r) < T_{[j]}(r) = \dots = T_{[j^\prime]}(r) < T_{[j^\prime+1]}(r)$, then the continuous rank is defined as $c_{[l]}(r) = \frac{j+j^\prime}{2} + \frac{1}{2}$ for $l = j, \dots, j^\prime$. 
The continuous ranks were first suggested by \citet{Hahn2015}, and are defined here with a slight modification in the scaling of $c_{[0]}(r)$ and $c_{[J]}(r)$.

Finally, the univariate Monte Carlo test is performed based on $c_j$. Thus the final $p$-value is $p^\text{cont}= \frac{1}{J+1} \sum_{j=0}^J \mathbb I (c_j \leq c_0)$.

\subsubsection{Area rank}
The last refinement is the area rank measure $a_j$ defined as: 
\begin{equation}\label{aj}
a_j  = \frac{1}{J+1}\left( R_j - \frac{1}{n}\sum_{r} (R_j-{C_{j}(r)}) \1({C_{j}(r)} < R_j) \right)
\end{equation}
where $n$ is equal to the number of $r \in I$. 
{Thus, the area measure refines the extreme rank $R_i$ by reducing from it the scaled sum of the differences between $R_j$ and $C_j(r)$ across the spatial points $r$ where the minimal $p$-values (or ranks) are obtained. Thus, the larger the amount of those critical $r$ and the smaller the pointwise continuous ranks $C_j(r)$ at those critical $r$ are, the smaller is the value of the area measure, i.e.\ the more extreme the corresponding test statistic $T_j(r)$.}
The univariate Monte Carlo test is performed based on $a_j$ with $p^\text{area}= \frac{1}{J+1} \sum_{j=0}^J \mathbb I (a_j \leq a_0)$.

\subsection{Illustration of the {new} measures}\label{sec:measuresexample}

Figure \ref{fig:measures} shows a small set of one-dimensional functions (test statistics) $T(r)$ for description of the behaviour of the new measures.
Only large values of $T(r)$ are considered significant.
In this set of functions, the variance of the distribution of {$T(r)$ increases along $r$ (from left to right)}. 
According to the $F$-max method the most extreme test statistic is the function 1 reaching the overall highest value across the locations $r$. 
On the other hand, the $p$-min method gives the minimal $p$-value to six functions (1, 3, 4, 5, 7 and 9). 
The refinement measures decide the extremeness of these functions by different rules.

The most extreme function according to the ERL measure is the function 5 (yellow), because it obtains {the largest value, i.e.} the minimal $p$-value, over the largest domain.
On the other hand, according to the continuous and area rank measures, the most extreme function is the function 9 (blue): 
The deviation of the functions 3 (red) and 9 (blue) from the second most extreme functions is of the same size at those locations $r$ where these functions reach the minimal $p$-value. However, the scaling of the continuous rank by the overall variation (see Eq.\ \eqref{eq:extremecontranks}) leads to smaller pointwise continuous ranks for the function 9, therefore also to a smaller overall continuous rank.
Because the two functions reach the minimal $p$-values approximately at domains of the same size, also the area measure deduces the function 9 as the most extreme function.

This example illustrated two issues:
1) The fact that the ERL, continuous and area measures are all refinements of the $p$-min method allows them to regard as most extreme also functions which reach their minimal $p$-value, $p_j^\pmin$, at locations where the variation of the functions is small. 
2) The refinement measures break the ties by different rules as summarized in Table \ref{tab:inf}.

 \begin{figure}[ht!]
     \centering
     \includegraphics[width=\textwidth]{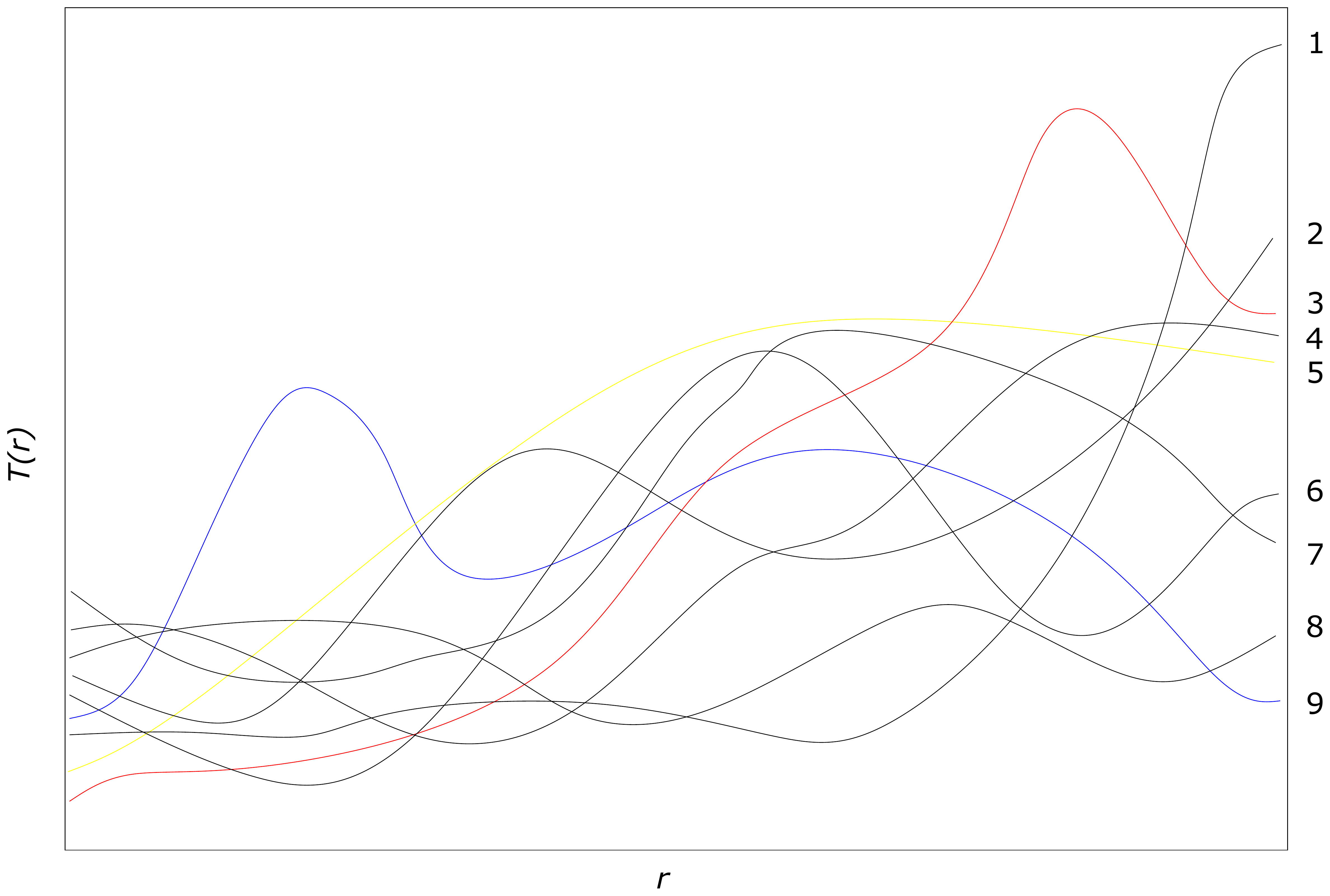}
\caption{A set of one-dimensional functions (test statistics) for description of the different measures. The variance of the distribution of the test statistics increases from left to right. The most extreme function according to ERL measure is the function 5, according to continuous and area rank measures it is the function 9. The function 3 is not as extreme according to these measures because the overall variation near its peak is higher than for the function 9.}
     \label{fig:measures}
 \end{figure}

\begin{table}[ht]
\centering
\begin{tabular}{cccc}
  \hline
& $p_j^\pmin$ & Domain with $p_j^\pmin$  & Value \\
  \hline
$F$-max  & No & No & Yes\\
$p$-min  & Yes & No & No  \\ 
 ERL   & Yes & Yes & No \\ 
 Cont    & Yes & No & Yes \\ 
 Area    & Yes & Yes & Yes \\ 
   \hline
\end{tabular} 
	\caption{Sources of information gathered by different measures. While $F$-max depends on the maximum value of the test statistic $T(r)$ across all $r\in I$, the other measures utilize the values of $T(r)$ at the locations $r$ with $p_j^\pmin$.
}\label{tab:inf}
\end{table}

\subsubsection{Graphical interpretation - $100\cdot (1-\alpha)$\% global envelope}

Consider any of the three measures we defined and denote them using a common notation $M_j, j=0, \ldots, J$. Let $I_\alpha = \{j\in 0,\dots, J: M_j \geq M_{(\alpha)} \}$ be the index set of {the test statistics}, where the threshold $M_{(\alpha)}$ is chosen to be the largest value in $\{M_0,\dots,M_{J}\}$ for which
\begin{equation}\label{eq:Ralpha}
  \sum_{j=0}^{J} \mathbb I \left(M_j < M_{(\alpha)}\right) \leq \alpha (J+1),
\end{equation}
{i.e.\ the set $I_\alpha$ contains the indices of the $100(1-\alpha)$\% of the least extreme test statistics $T_j(r)$. The global envelope is defined to be}
\begin{equation*}%\label{rank_envelopes}
  T_{\lo}^{(\alpha)}(r)= -\infty, \quad   T_{\up}^{(\alpha)}(r)= \max_{j\in I_\alpha}T_{j}(r).
\end{equation*}
As a consequence, the probability that $T_0(r)$ falls outside this global envelope in any of the $r$ points is less or equal to $\alpha$, 
\begin{equation*}%\label{eq:envelope}
  \Pr\big(\exists\, r\in I: T_0(r)\notin[T_{\lo}^{(\alpha)}(r),T_{\up}^{(\alpha)}(r)])\leq \alpha.
\end{equation*}

The following theorem states that inference based on the $p$-value $p^M$ (where $p^M$ stands for $p^\text{erl}, p^\text{cont}$ and $p^\text{area}$) and on the global envelope specified by $T_{\lo}^{(\alpha)}(r)$ and $T_{\up}^{(\alpha)}(r)$ with respect to the appropriate measure are equivalent. Therefore, we can refer to these envelopes as the $100\cdot(1-\alpha)$\% global \emph{extreme rank length envelope}, \emph{continuous rank envelope} and \emph{area rank envelope}.
Figure \ref{fig:upenv} shows an example of the extreme rank length upper envelope together with red points indicating the area where the test function crossed this envelope. Thus the graphical interpretation identifies the voxels which cause the rejection.

\begin{theorem}\label{thm:envelope-vs-pinterval}
Assume that there are no pointwise ties with probability $1$ among $T_j(r), j=0,\ldots , J$. Then
\begin{enumerate}
 \item $T_0(r) >T_{\up}^{(\alpha)}(r)$ for some $r \in I$
iff $p^M \leq \alpha$, in which case the null hypothesis is rejected;
 \item $T_0(r) \leq T_{\up}^{(\alpha)}(r)$ for all $r \in I$ iff $p^M > \alpha$, and thus the null hypothesis is not rejected;
\end{enumerate}
\end{theorem}
\begin{proof}
Since 1.\ holds iff 2.\ holds, it's enough to show 1.
According to the definition of $p^M$, $p^M \leq \alpha$ iff number of  $M_j$ smaller or equal to $M_0$ is smaller or equal to $\alpha (J+1)$. That is equivalent, according to the definition of $M_{(\alpha)}$ to $M_0 < M_{(\alpha)}$. This holds iff $0 \notin I_\alpha$, which is equivalent to  $T_0(r) >T_{\up}^{(\alpha)}(r)$ for some $r \in I$ according to the definition of the measure $M$.

 \end{proof}

\subsection{$F$-max approach }
Usually the statistic $T(r)$ is one of the statistics used in the parametric general linear models, i.e.\ $t$-statistic, $F$-statistic, or variance weighted $F$-statistic. In these cases, an immediate solution to the problem of ties is to replace $p_j^\text{min}$ statistic in the Monte Carlo test directly by the $F_j^\text{max}=\max_{r \in I} T_j(r)$ statistic, but with the price of loosing the pivotality and, consequently also the power. See the next section for details.

\section{Simulation experiment}
\label{se:se}
To compare the power and robustness of the proposed methods and the existing multiple comparison methods under different scenarios, we generated synthetic imaging data mimicking real data from neuroimaging studies. 
We considered a categorical factor $g$ taking the values 1 or 2 according the group to which the image belongs to,
and a continuous factor $z$ that was generated from the uniform distribution on $(0, 1)$. 
We simulated images $Y(r)$ in the square window $[-1, 1]^2$ on a grid of $51 \times 51$ pixel resolution from the following GLM models: 
$$
\text{M0}:\quad Y(r) = \epsilon(r)
$$
\begin{eqnarray*}
&& \text{M1}: \quad Y(r) = \exp(-10\cdot ||r||) \cdot g + \epsilon(r) \\
&& \text{M1}': \quad Y(r) = \exp(-200\cdot ||r||) \cdot g + \epsilon(r)
\end{eqnarray*}
and
$$
\text{M2}:\quad Y(r) = \exp(-10\cdot ||r||) \cdot (g + z) + \epsilon(r).
$$
The image resolution is obviously small in comparison the brain data (see Section \ref{sec:introexample}), but was regarded as sufficient for the designed experiment to show differences between the measures.
Here $||r||$ denotes the Euclidean distance of the pixel to the origin and $\epsilon(r)$ is a zero-mean correlated error. The model M0 has no factors and generates purely noisy images, the models M1 and M1' generate two groups of images depending on the categorical factor $g$ and the model M2 generates images which depend on both the categorical and continuous factors. 
The models M1 and M1' correspond to simple comparison of two groups of images, where the "bump" at the center of the image is two times higher for second group than for the first group. 
The area where the two groups differ is about hundred times smaller in the model M1' than in M1. 
Because in M1 the departures from the null model occur for many pixels $r$ and the ERL and Area methods accumulate information across the pixels,
the ERL and Area methods are expected to have advantage over the other methods in this model. On the other hand, the departures from the null model are expected only for a few pixels $r$ in M1' and, thus, such advantage does not occur in this model.

The model M0 is a null model where the images in the two groups are from the same model (no factors) and it was used for estimating the significance levels of the different tests. The model M2 is similar to M1, but the groups are disturbed by the continuous factor $z$. The models {M1, M1' and M2} were used for power estimation. The permutation of raw data was used in models M1 and M1'. The Freedman and Lane permutation scheme was used in M2. Because the model M2 revealed consistent results with model M1, the results are not presented.

For both models M1 and M1', seven different correlated error structures $\epsilon(r)$ were considered:
\begin{itemize}
\item[(a)] Gaussian error $\epsilon(r) = G_{0.15}(r)$,  where $G_{\rho}(r)$ is a Gaussian random field with the exponential correlation structure with scale parameter $\rho$ and standard deviation $\sigma$ which will take several values,
\item[(b)] skewed error $\epsilon(r) = \exp(G_{0.15}(r))$, 
\item[(c)] shape inhomogeneous distribution with increasing bimodality \\ 
$\epsilon(r) = \frac{1}{4}\text{sign}(G_{0.15}(r))|G_{0.15}(r)|^{1/(2\|r\|+1)}$, 
\item[(d)] 
$
\epsilon(r) = \mathbb I (\|r\| > 0.5) \frac{1}{2}{G_{0.15}(r)^{1/5}} + \mathbb I (\|r\| \leq 0.5) G_{0.15}(r)
$,
i.e.\ inhomogeneous distribution over $I$ with the bimodal and normal errors in the periphery and in the middle of the image, respectively,
\item[(e)] 
$
\epsilon(r) = \frac{1}{8}\mathbb I (\|r\| \leq 0.5) \exp(3G_{0.15}(r)) + \frac{1}{8}\mathbb I (\|r\| > 0.5) (G_{0.15}(r)+1) 
$,
i.e.\ inhomogeneous distribution over $I$ with the skewed and normal errors in the middle and periphery of the image, respectively,
\item[(f)]
$
\epsilon(r) = \mathbb I (\|r\| \leq 0.5) G_{0.05}(r) + \mathbb I (\|r\| > 0.5) G_{0.3}(r)
$,
i.e.\ Gaussian distribution with inhomogeneity in the correlation structure (scale parameters 0.05 and 0.3 in the middle and periphery of the image), 
\item[(g)]
$
\epsilon(r) = \mathbb I (\|r\| \leq 0.5)\frac{1}{2} G_{0.05}(r)^{1/5} + \mathbb I (\|r\| > 0.5) \frac{1}{2}G_{0.3}(r)^{1/5}
$,
i.e.\ bimodal distribution with inhomogeneity in the correlation structure.
\end{itemize}
The homogeneous error distributions (a) and (b) represent cases where all methods should perform well in identifying alternative hypothesis.
The skewed error (b) represents a situation where the permutation inference is necessary because the assumptions needed for parametric methods are not met. 
Further, the errors (c), (d) and (e) are inhomogeneous across the image and illustrate cases where the variability of the test statistic $T(r)$ in the periphery of the image mask the signal from the center of the image. The error (c) is rather normal in the center of the image while the increasing bimodality appears closer to the periphery. The bimodal error corresponds to the sharp changes in the images, whereas Gaussian error corresponds to the smooth changes in the images. The errors (d) and (e) have normal error at different places, in order to have bigger quantiles of the test statistic at the periphery of the image.  
On the other hand, the errors (f) and (g) are used to investigate the effect of inhomogeneity in the correlation structure on the results.

All the images in the simulation study had the resolution $51 \times 51$ pixels.
The number of permutations used throughout the whole simulation study was set to 2000.
The estimated significance levels and powers were recorded for various cases for all studied multiple comparison methods, i.e.\ for $F$-max, $p$-min and the proposed ERL, continuous and area rank methods.

To capture the behavior of the methods for various levels of significance, seven different standard deviations {$\sigma_1 = 0.1$, $\sigma_2 = 0.2$, $\sigma_3 = 0.3$, $\sigma_4 = 0.5$, $\sigma_5 = 0.75$, $\sigma_6 = 1$ and $\sigma_7 = 1.25$} were used in all studied cases. Each model was simulated with ten images per group and each experiment was repeated 1000 times in order to obtain estimated significance levels and powers. 

Tables \ref{M0}, \ref{M1} and \ref{M1'} show the results for models M0, M1 and M1' in a shorted way. The detailed results are presented in Appendix C. 

The estimated significance levels revealed the same structure for all errors (a)-(g) (Table \ref{M0}): The $p$-min was enourmously conservative. The $F$-max and all our proposed methods achieved the preset significance level $\alpha=0.05$ in all cases. (The 95\% confidence interval for 1000 simulations with success probability 0.05 is (0.037, 0.063).) The only exception was the error (c), where Cont was slightly conservative. 

The conservative $p$-min method had no power (Tables \ref{M1} and \ref{M1'}).
The ERL and Area methods had uniformly higher power than $F$-max for all errors in case of the model M1. This is caused by the fact that ERL and Area summarize the extremeness of the test statistic from all spatial points. This advantage led to higher power even in the cases {(f)-(g) with inhomogeneous autocorrelation}.

The situation is more complicated for model M1', where  the area of extremeness is rather small: the ERL and Area can not benefit from their feature of collecting information from all spatial points. For the homogeneous cases with normal error (a) and skewed error (b), the  $F$-max is slightly more powerful than our proposed methods. When the inhomogeneity of errors is set in such a way that the {values of the test statistic quantiles are larger} at periphery where no differences are present -- cases (c), (d) and (e) -- the ERL and Area methods are significantly more powerful than $F$-max (Table \ref{M1'}). Finally, when {the range of correlation is larger} at the periphery where no differences are present, the methods accumulating information from neighborhood should be affected by this fact. Really, in case (f) ERL appears to be less powerful than $F$-max, but the Area method which uses {different} sources of information {(see Table \ref{tab:inf})} is affected by this fact only little. This decrease of power is not further seen in case (g) where the normal distribution is replaced by the bimodal distribution. At last, we observe that Cont and $F$-max methods were rather equivalent in our study with 2601 spatial points and 2000 permutations.

\begin{table}[ht]
\centering
\begin{tabular}{rrrrrrrr}
\hline
& M0a-$\sigma_ 1$ & M0b-$\sigma_ 1$ & M0c-$\sigma_ 1$ & M0d-$\sigma_ 1$ & M0e-$\sigma_ 1$ & M0f-$\sigma_ 1$ & M0g-$\sigma_1$ \\
\hline
F-max   & 0.054 & %0.0555
0.056 & 0.040 & 0.046 & 0.053 & 0.058 & 0.043 \\ 
$p$-min   & 0.000 & 0.000 & 0.000 & 0.000 &0.000& 0.000 & 0.000 \\ 
%FDR   & 0.000 & 0.000 & 0.000 & 0.000 & 0.002 & 0.000 \\ 
%$p$-min-rand   & 0.044 & 0.047 & 0.059 & 0.054 & 0.047 & 0.042\\ 
ERL   & 0.060 & 0.052 & 0.045 & 0.043 & 0.047& 0.057 & 0.045  \\ 
Cont   & 0.058 & 0.056 & 0.031 & 0.043 & 0.045 & 0.049 & 0.053 \\ 
Area   & 0.058 & 0.047 & 0.040 & 0.040 & 0.043& 0.061 & 0.045  \\
\hline
\end{tabular}
\caption{\label{M0}Empirical significance levels of all studied methods {based on 1000 replicates} for model M0 with error (a)-(g) with standard deviation $\sigma_1=0.1$. }
\end{table}

\begin{table}[ht]
\centering
\begin{tabular}{rrrrrrrr}
\hline
&M1a-$\sigma_ 5$ & M1b-$\sigma_ 5$ & M1c-$\sigma_ 4$ & M1d-$\sigma_ 4$ & M1e-$\sigma_ 7$ & M1f-$\sigma_ 5$ & M1g-$\sigma_ 6$ \\
\hline
F-max    & 0.395 & 0.464 & 0.267 & 0.334 & 0.294 & 0.632 & 0.421 \\ 
$p$-min    & 0.000 & 0.000 & 0.000 & 0.000 & 0.000 & 0.000 & 0.000 \\ %FDR    & 0.062 & 0.314 & 0.320 & 0.570 & 0.001 & 0.523 \\ 
%$p$-min-rand  & 0.072 & 0.099 & 0.093 & 0.076 & 0.090 & 0.084 \\ 
ERL    & 0.606 & 0.594 & 0.948 & 0.943 & 0.751 & 0.776 & 0.984 \\ 
Cont   & 0.378 & 0.399 & 0.228 & 0.363 & 0.335 & 0.522 & 0.358 \\ 
Area  & 0.544 & 0.568 & 0.879 & 0.909 & 0.577 & 0.749 & 0.979 \\
\hline
\end{tabular}
\caption{\label{M1}Empirical powers of all studied methods {based on 1000 replicates} for model M1 with error (a)-(f). The chosen standard deviation $\sigma$ for each case is the one with maximal contrast between methods.}
\end{table}

\begin{table}[ht]
\centering
\begin{tabular}{rrrrrrrr}
\hline
& M1'a-$\sigma_ 3$ & M1'b-$\sigma_ 2$ & M1'c-$\sigma_ 7$ & M1'd-$\sigma_ 3$ & M1'e-$\sigma_ 4$ & M1'f-$\sigma_ 3$ & M1'g-$\sigma_ 1$ \\
\hline
F-max   & 0.917 & 0.562 & 0.377 & 0.549 & 0.190 & 0.963 & 0.412 \\ 
$p$-min   & 0.000 & 0.000 & 0.000 & 0.000 & 0.000 & 0.000 & 0.000 \\ 
%FDR   & 0.000 & 0.000 & 0.000 & 0.000 & 0.002 & 0.000 \\ 
%$p$-min-rand   & 0.081 & 0.099 & 0.077 & 0.067 & 0.113 & 0.088 \\ 
ERL   & 0.732 & 0.365 & 0.972 & 0.752 & 0.235 & 0.720 & 0.894 \\ 
Cont   & 0.799 & 0.461 & 0.339 & 0.445 & 0.221 & 0.899 & 0.408 \\ 
Area   & 0.825 & 0.500 & 0.915 & 0.800 & 0.259 & 0.924 & 0.924  \\
\hline
\end{tabular}
\caption{\label{M1'}Empirical powers of all studied methods {based on 1000 replicates} for model M1' with error (a)-(f). The chosen standard deviation $\sigma$ for each case is the one with maximal contrast between methods. }
\end{table}

\section{Computational details: how to implement the methods for large brain data}\label{sec:computationalissues}

The $p$-min, ERL, Cont and Area measures are based either on ordinary or continuous pointwise ranks of the test statistics $T_j(r)$, $J=0,1,\dots,J$, among each other. A naive algorithm would first calculate all the ranks and save them in the memory in order to calculate the measures thereafter. However, the space complexity of the naive algorithm is $O(Jn)$, and for large data all the ranks may not fit comfortably to the memory of an ordinary computer. 
However, for all the proposed measures, it is possible to split the task by the locations (voxels) $r\in I$. In this case, the space complexity is just $O(Js)$, where $s$ is the number of subjects.
We note that the space complexity can be even further reduced to $O(J)$ if the permutations are not saved. This can be achieved by resetting the seed of the pseudo random number generator for each split. This is typically not necessary.

There is some extra computational complexity of the rank based multiple comparison in comparison to the $F$-max method because the test statistics need to be ranked. 
Namely, letting $S$ denote the time complexity of the test statistic to be computed, 
the complexity of computing rank based multiple comparison correction is $O(JnS + Jn\log J)$, while for $F$-max it is $O(JnS)$. Thus, asymptotically the rank based corrections are slower than $F$-max. This is the same whether the naive or location-wise calculations are used.

Algorithm \ref{alg:computation} presents the algorithm for location-wise (voxel by voxel) calculation of the rank-based measures. 
The update rules and calculation of the final measures are described thereafter for the different measures in Sections \ref{sec:update_pminCont}, \ref{sec:update_Area} and \ref{sec:update_ERL}.
The implementation is obvious for the $p$-min and Cont measures, while for ERL and Area, some more complicated calculations are needed. 
The computations are implemented in the function partial\_forder of the R package GET \citep{MyllymakiMrkvicka2020}.
{Section \ref{sec:compenvelope} explains how the global envelope can be calculated given the measures.}

\begin{algorithm}\label{alg:computation}
\caption{An algorithm to compute functional ordering voxel by voxel. Below {\em ranks} refers to the pointwise ranks, either ordinary (for $p$-min and ERL) or continuous (for Cont and Area). The {\em update} rules as well as final calculations are described in Section \ref{sec:update_pminCont} for $p$-min and Cont, in Section \ref{sec:update_Area} for Area and in Section \ref{sec:update_ERL} for ERL.}
\begin{algorithmic}
 \State Initialize the (augmented) measure $M_j$ for $j=0,\dots,J$ 
 \State Generate $J$ permutations
 \For{each voxel}
  \State $T_0 \gets $ ($F$-)statistic for data
  \For{$j \gets 1,\dots,J$}
   \State $T_j \gets $ ($F$-)statistic for permutation $j$
  \EndFor
  \State $(m_0, m_1, \dots, m_J) \gets ranks(T_0, T_1\dots, T_J)$
  \For{$j \gets 0,\dots,J$}
   \State $M_j \gets update(M_j, m_j)$
  \EndFor
 \EndFor
 \State Compute the final measures from the augmented measures
\end{algorithmic}
\end{algorithm}

\subsection{Update rule for $p$-min and Cont}\label{sec:update_pminCont}

For the $p$-min and Cont measures, the update rule is simple:
\begin{equation*}%\label{eq:update_pmin}
M_j \gets \min(M_j, m_j).
\end{equation*}
When all locations $r\in I$ have been handled, the final measure, either $p^\pmin_j$ or $c_j$, is obtained by dividing the $M_j$ by $J+1$.

\subsection{Update rule for Area}\label{sec:update_Area}

For the Area and ERL measures, it is necessary to augment the measure with some auxiliary information during the computation.
Namely, for the Area measure, the extreme rank $R_j$ and the difference between the extreme rank $R_j$ and the pointwise continuous rank $C_j(r)$ has to be saved and updated for data and each permutation $j$. Let $D_j$ denote the difference.
Initially $M_j = (R_j, D_j) = (\infty, 0)$.
For the update there are three possibilities:
\begin{equation*}%\label{eq:update_area}
    M_j = (R_j, D_j) \gets \begin{cases}
     (R_j, D_j) & \text{ if } ceil(m_i) > R_i\\
     (R_j, D_j + ceil(m_j) - m_j) & \text{ if } ceil(m_j) = R_j\\
     (ceil(m_j), ceil(m_j) - m_j) & \text{ if } ceil(m_j) < R_j
    \end{cases}
\end{equation*}
The final measure is 
\begin{equation*}%\label{eq:final_area}
a_j = \frac{1}{J+1}(R_j - D_j/n).
\end{equation*}

\subsection{Update rules for ERL}\label{sec:update_ERL}

For ERL, the augmented measure $M_j$ is a vector of the chosen number of most extreme pointwise ranks of $T_j(r)$ and their counts among $r\in I$. 
We record the six most extreme ranks, thus let
$$
M_j = (R_1,R_2,\dots, R_6, c_1, c_2, \dots, c_6) = (\infty, \dots, \infty, 0, \dots, 0)
$$
be the initilized vector. Then the updates are obtained by the following rules:
\begin{itemize}
\item If $m_j = R_i$ for some $i = 1, \dots, 6$,
then $c_i \gets c_i + 1$.
\item If $R_{i-1} < m_j < R_i$ for some $i=1,\dots,6$, then 
$$R_{k} \gets R_{k-1}, c_k \gets c_{k-1} \text{ for } k = i+1, \dots, 6$$ and $R_i \gets m_j, c_i \gets 1$.
\end{itemize}
The final measure for ERL is the ranking of the augmented measures $M_j$ according to a special ordering that is defined as follows. Let $R_i, c_i$ and $R_i', c_i'$ be the augmented measures for two functions 1 and 2.
\begin{enumerate}
    \item Start with the most extreme rank $i \gets 1$
    \item if $R_i < R_i'$ then function 1 is more extreme.
    \item if $R_i > R_i'$ then function 2 is more extreme.
    \item if $R_i = R_i'$ and $c_i < c_i'$ then function 2 is more extreme.
    \item if $R_i = R_i'$ and $c_i > c_i'$ then function 1 is more extreme.
    \item if $R_i = R_i'$ and $c_i = c_i'$ then set $i \gets i + 1$ and go to step 2.
\end{enumerate}
The final normalized ERL measures \eqref{eq:lexicalrank} are obtained by dividing by $J+1$.

\subsection{Computation of the global envelopes}\label{sec:compenvelope}

After the final measures have been calculated and the $p$-value of the test is less than the chosen significance level, it is possible to construct the envelope to detect which parts of the domain have led to the rejection of the null hypothesis.
Algorithm \ref{alg:envelopes} describes the computation. 

\begin{algorithm}\label{alg:envelopes}
\caption{An algorithm to compute global envelopes after the functional orderings have been calculated.}
\begin{algorithmic}
 \State Let $M_j$ be the {final} measure for $j=0,\dots,J$ and $M^{(\alpha)}$ the critical value{, i.e.\ the largest of the $M_i$ such that the number of those $j$ for which $M_j< M^{(\alpha)}$ is less or equal to $\alpha (J+1)$}
 \State Use the same $J$ permutations as when computing the functional orderings
 \For{each voxel $v$}
  \State $T_0 \gets $ ($F$-)statistic for data
  \For{$j \gets 1,\dots,J$}
   \State $T_j \gets $ ($F$-)statistic for permutation $j$
  \EndFor
  \State $U_v \gets \max\{T_j : M_j > M^{(\alpha)}\}$
 \EndFor
 \State {$U_v$ contains the upper envelope.}
\end{algorithmic}
\end{algorithm}

\section{Discussion}
All the permutation based multiple testing procedures reach the prescribed significance level, if the permutation strategy leads to exchangeable test statistics. They even reach the prescribed level approximately in the presence of nuisance factors, if the Freedman and Lane permutation strategy is used. This fact was demonstrated here by a simulation study and the \citet{EklundEtal2016} type of experiment presented in Appendix A. 

The measures used in the permutation test are different in two ways, their sensitiveness towards different types of extremeness and their robustness against different types of inhomogeneities of the error term $\epsilon(r)$ across $r\in I$. 
The sensitivities can be understood from the sources of information which are gathered by the different measures (see Table \ref{tab:inf}): 
The integral type of extremeness is detected by those measures (ERL, Area) that gather information across the locations $r\in I$ with the minimal $p$-values, while the measures based on the value of extremes (Cont, Area) are sensitive to maximum type of extremeness.
The categorization is given by the construction of the measures, and it was further supported by the simulation study results for cases M1 and M1'.

Table \ref{tab:problem} summarizes the investigated measures and their problems with identification of the alternative hypothesis, i.e. with ties and various inhomogeneities of the error term. The inhomogeneities can decrease the power of the procedures as stated in the table.
As it was stated already in the introduction, the inhomogeneity and non normality of error term $\epsilon(r)$ across $r\in I$ together lead to inhomogeneous quantiles of the $F$-statistic. 
Therefore, the $F$-max procedure can be blind to the departures where the quantiles of $T(r)$ are less variable. This was demonstrated in Section \ref{sec:introexample} and in the simulation study by cases c), d) and e). Since the Cont measure depends also on the values of $T(r)$ and not only ranks, the Cont measure has the same problem as $F$-max with inhomogeneity of $T(r)$. However, the problem of Cont depends on the number of ties to break: the bigger their number is, the less sensitive Cont is expected to be to inhomogeneity.

On the other hand, {in the presence of the inhomogeneity of the correlation structure of $\epsilon(r)$ across $r\in I$, the ERL measure can be blind to the departures from the null that occur in areas where the range of correlation of $\epsilon(r)$ is smaller.}  This was demonstrated by the simulation study cases f) and g).  

The Area rank measure is a combination of Cont and ERL, and thus it should be sensitive to both inhomogeneities but as the simulation study shows the decrease of power in cases c), d), e), f) and g) was not as apparent as for the $F$-max, ERL or Cont procedures. Thus the Area measure appears to be the most robust method with respect to inhomogeneity of the error term $\epsilon(r)$.

Further, the simulation study case g) showed that the rank based measures are much more robust with respect to extreme non normality than the $F$-max test as the decrease of power between cases f) and g) is apparent only for $F$-max procedure.

\begin{table}[ht]
\centering
\begin{tabular}{cccc}
  \hline
  & Ties & Inhomogeneity and  & Inhomogeneity of correlation  \\ 
  &&non-normality of $\epsilon(r)$&structure of $\epsilon(r)$\\
  \hline
$F$-max  & No & Yes & No\\
$p$-min  & Yes & No & No  \\ 
 ERL   & No & No & Yes \\ 
 Cont    & No & Yes & No \\ 
 Area    & No & Partial & Partial \\ 

  \hline
\end{tabular} 
	\caption{Problems of different measures.}\label{tab:problem}
\end{table}

\section{Conclusions}

We presented three new multiple comparison methods for the permutation inference for the GLM.
We showed by a designed experiment that 
the proposed methods have the desired significance level unlike the $p$-min method that was highly conservative, unless enormous amount of permutations was used.
The choice of the method from $F$-max, ERL, Cont and Area rank measures should depend on the assumptions about error term $\epsilon(r)$ that can be made and the expected type of extremeness of the test statistic.
{Both homogeneity of the distribution of $\epsilon(r)$ (HD) and homogeneity of the correlation structure of $\epsilon(r)$ (HCS) across the image play a role (see Table \ref{tab:problem}). When we assume
\begin{itemize}
    \item both HD and HCS, then all methods can be applied. The ERL or Area methods should be preferred when the test statistic is expected to be extreme over a large area of the image $I$; $F$-max, Cont or Area when extremeness is expected only on a small area. The same recommendations hold, if the assumption of homogeneity of $\epsilon(r)$ is replaced by the assumption of the normality of the error $\epsilon(r)$.
    \item HCS but not HD, then only ERL can be applied without worry of loosing power to identify alternative hypotheses. However, in our designed experiment, Area showed good robustness with respect to inhomogeneity of the distribution of $\epsilon(r)$, whereas $F$-max and Cont did not.
    \item HD but not HCS, then $F$-max and Cont can be applied without worry of loosing power to identify alternative hypotheses. Again Area showed good robustness with respect to inhomogeneity of the correlation structure of $\epsilon(r)$, whereas ERL did not.
    \item neither HD nor HCS, then none of the methods can be applied without worry of loosing power to identify alternative hypotheses, but Area showed good robustness with respect to both types of inhomogeneities.
\end{itemize}}

Thus, our experience suggests that the Area method can be used in a general situation without worry of loosing much power. Its dual nature of being based both on extremeness of the test statistic as well as on summarizing this extremeness across spatial points makes it sensitive to different departures from the null model and robust to different kinds of inhomogeneity. 
{This was demonstrated by a designed experiment as well as an example of a real data study, where the Area method detected significant differences between groups while $F$-max did not. 
From these reasons, we believe that the ERL and Area methods can help to find further understanding of the phenomena and brain function and, thus, prove to be useful for the neuroimaging practice.}

\section*{Acknowledgements}
T.\ Mrkvi\v cka has been financially supported by the Grant Agency of Czech Republic (Project No. 19-04412S), and M.\ Myllym{\"a}ki and M.\ Kuronen by the Academy of Finland (project numbers 295100, 306875, 327211). 
The authors wish to acknowledge CSC -- IT Center for Science, Finland, for computational resources.

\section*{Data availability}
We confirm that the used data of autism brain imaging from ABIDE are freely available at http://fcon\_1000.projects.nitrc.org/indi/abide/abide\_I.html.

\section*{References}

%\bibliographystyle{}
%\bibliography{Tomas_bibfile} 

\section*{Appendix A: Study on false discoveries}

To study the significant levels of the ERL, Cont and Area tests by an experiment similar to those proposed by \citet{EklundEtal2016}, we used the measurements of the 573 typical controls of the autism brain imaging data of Section \ref{sec:introexample}. For this experiment, we chose the left and right Crus Cerebelum 1 regions of the brain containing 4783 voxels. 
We took 1000 samples of 20 control subjects and each sample was randomly split in two equally sized groups.
For each sample, we tested the effect of the group in the model \eqref{eq:GLM} where age was a nuisance factor, using 2499 permutations by the Freedman-Lane algorithm \citep{FreedmanLane1983}. 
This number of permutations was chosen to have the number of permutations about half of the number of voxels as in the illustrative example of Section \ref{sec:introexample}. 
All the empirical significance levels were $0.040$.

\section*{Appendix B: The global envelope for the brain data example}

Figure \ref{fig:ABIDEArea2} shows the 95\% global Area envelope and the significant voxels for the whole brain data for the example discussed in Section \ref{sec:introexample}.
The number of simulations was 100000 and $p$-value 0.040 for the test of the effect of the group.

\begin{figure}
    \centering
    \includegraphics[width=\textwidth]{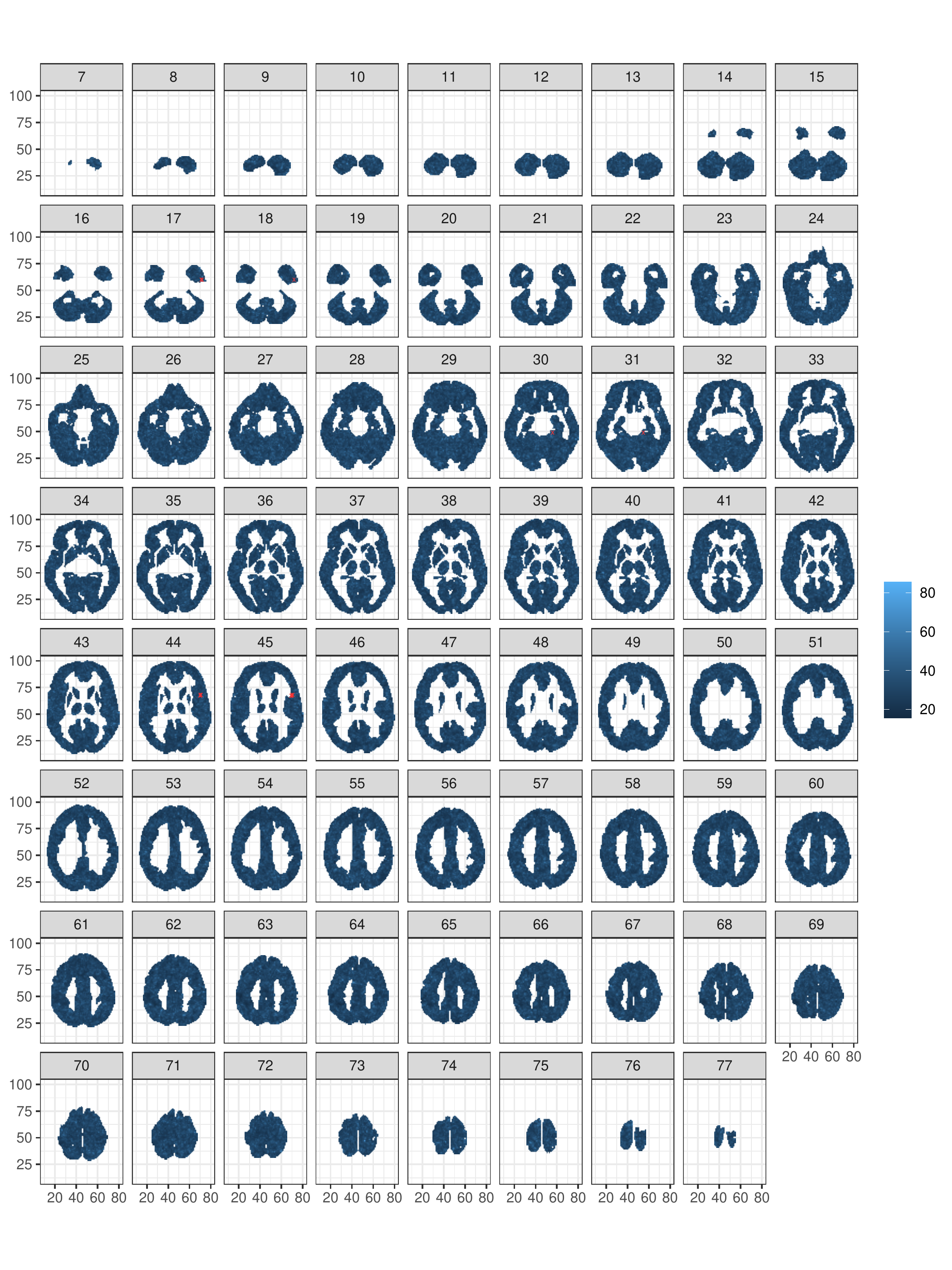}
    \caption{The upper 95\% envelope of the voxelwise $F$-statistic of the studied model where the type (ASD or TC) was the interesting factor (blue image) together with the significant voxels detected by the area rank test (red crosses). The brain is shown by 71 2D slices (with depths varying from 7 to 77).}
    \label{fig:ABIDEArea2}
\end{figure}

\section*{Appendix C: Further results for the simulation study}

The following tables contain all results of the simulation study for M1 and M1'.

\begin{table}[ht]
\centering
\begin{tabular}{rrrrrrrr}
  \hline
 & $\sigma_ 1$ & $\sigma_ 2$ & $\sigma_ 3$ & $\sigma_ 4$ & $\sigma_ 5$ & $\sigma_ 6$ & $\sigma_ 7$ \\
  \hline
F-max   & 0.054 & 0.048 & 0.043 & 0.051 & 0.044 & 0.045 & 0.047 \\ 
  $p$-min   & 0.000 & 0.000 & 0.000 & 0.000 & 0.000 & 0.000 & 0.000 \\ 
%  FDR   & 0.000 & 0.000 & 0.000 & 0.000 & 0.000 & 0.000 & 0.000 \\ 
%  $p$-min-rand   & 0.044 & 0.051 & 0.051 & 0.050 & 0.051 & 0.045 & 0.040 \\ 
  ERL   & 0.060 & 0.036 & 0.042 & 0.048 & 0.052 & 0.046 & 0.040 \\ 
  Cont   & 0.058 & 0.060 & 0.047 & 0.042 & 0.052 & 0.055 & 0.047 \\ 
  Area   & 0.058 & 0.040 & 0.041 & 0.053 & 0.047 & 0.050 & 0.040 \\ 
   \hline
\end{tabular}
\caption{\label{M0a}Estimated significance levels of all studied methods and 7 standard deviations for model M0 with error (a).}
\end{table}

\begin{table}[ht]
\centering
\begin{tabular}{rrrrrrrr}
  \hline
 & $\sigma_ 1$ & $\sigma_ 2$ & $\sigma_ 3$ & $\sigma_ 4$ & $\sigma_ 5$ & $\sigma_ 6$ & $\sigma_ 7$ \\
  \hline
F-max   & 1.000 & 1.000 & 1.000 & 0.863 & 0.395 & 0.184 & 0.106 \\ 
  $p$-min   & 0.000 & 0.000 & 0.000 & 0.000 & 0.000 & 0.000 & 0.000 \\ 
%  FDR   & 1.000 & 1.000 & 1.000 & 0.540 & 0.062 & 0.009 & 0.002 \\ 
%  $p$-min-rand   & 0.094 & 0.080 & 0.079 & 0.072 & 0.063 & 0.077 & 0.057 \\ 
  ERL   & 0.979 & 0.975 & 0.980 & 0.943 & 0.606 & 0.292 & 0.188 \\ 
  Cont   & 0.979 & 0.975 & 0.980 & 0.824 & 0.378 & 0.170 & 0.105 \\ 
  Area   & 0.979 & 0.975 & 0.980 & 0.929 & 0.544 & 0.257 & 0.150 \\ 
   \hline
\end{tabular}
\caption{\label{M1a}Empirical powers of all studied methods and 7 standard deviations for model M1 with error (a).}
\end{table}

\begin{table}[ht]
\centering
\begin{tabular}{rrrrrrrr}
  \hline
 & $\sigma_ 1$ & $\sigma_ 2$ & $\sigma_ 3$ & $\sigma_ 4$ & $\sigma_ 5$ & $\sigma_ 6$ & $\sigma_ 7$ \\    \hline
F-max   & 1.000 & 1.000 & 1.000 & 0.886 & 0.464 & 0.242 & 0.151 \\ 
  $p$-min   & 0.000 & 0.000 & 0.000 & 0.000 & 0.000 & 0.000 & 0.000 \\ 
 % FDR   & 1.000 & 1.000 & 1.000 & 0.534 & 0.060 & 0.009 & 0.001 \\ 
 % $p$-min-rand   & 0.092 & 0.091 & 0.079 & 0.105 & 0.086 & 0.055 & 0.066 \\ 
  ERL   & 0.973 & 0.980 & 0.978 & 0.942 & 0.594 & 0.325 & 0.233 \\ 
  Cont   & 0.973 & 0.980 & 0.978 & 0.844 & 0.399 & 0.215 & 0.163 \\ 
  Area   & 0.973 & 0.980 & 0.978 & 0.927 & 0.568 & 0.303 & 0.215 \\ 
   \hline
\end{tabular}
\caption{\label{M1f}Empirical powers of all studied methods and 7 standard deviations for model M1 with error (b). }
\end{table}

\begin{table}[ht]
\centering
\begin{tabular}{rrrrrrrr}
  \hline
 & $\sigma_ 1$ & $\sigma_ 2$ & $\sigma_ 3$ & $\sigma_ 4$ & $\sigma_ 5$ & $\sigma_ 6$ & $\sigma_ 7$ \\
  \hline
F-max   & 1.000 & 1.000 & 0.959 & 0.267 & 0.097 & 0.067 & 0.058 \\ 
  p-min   & 0.000 & 0.000 & 0.000 & 0.000 & 0.000 & 0.000 & 0.000 \\ 
  ERL   & 0.986 & 0.976 & 0.979 & 0.948 & 0.520 & 0.269 & 0.175 \\ 
  Cont   & 0.986 & 0.976 & 0.945 & 0.228 & 0.090 & 0.065 & 0.050 \\ 
  Area   & 0.986 & 0.976 & 0.979 & 0.879 & 0.395 & 0.191 & 0.121 \\ 
 \hline
\end{tabular}
\caption{\label{M1c}Empirical powers of all studied methods and 7 standard deviations for model M1 with error (c).}
\end{table}

\begin{table}[ht]
\centering
\begin{tabular}{rrrrrrrr}
  \hline
 & $\sigma_ 1$ & $\sigma_ 2$ & $\sigma_ 3$ & $\sigma_ 4$ & $\sigma_ 5$ & $\sigma_ 6$ & $\sigma_ 7$ \\
  \hline
F-max   & 1.000 & 1.000 & 0.989 & 0.334 & 0.080 & 0.054 & 0.048 \\ 
  p-min   & 0.000 & 0.000 & 0.000 & 0.000 & 0.000 & 0.000 & 0.000 \\ 
  ERL   & 0.980 & 0.986 & 0.982 & 0.943 & 0.617 & 0.299 & 0.199 \\ 
  Cont   & 0.980 & 0.986 & 0.958 & 0.363 & 0.097 & 0.058 & 0.054 \\ 
  Area   & 0.980 & 0.986 & 0.982 & 0.909 & 0.492 & 0.229 & 0.157 \\  
   \hline
\end{tabular}
\caption{\label{M1d}Empirical powers of all studied methods and 7 standard deviations for model M1 with error (d).}
\end{table}

\begin{table}[ht]
\centering
\begin{tabular}{rrrrrrrr}
  \hline
 & $\sigma_ 1$ & $\sigma_ 2$ & $\sigma_ 3$ & $\sigma_ 4$ & $\sigma_ 5$ & $\sigma_ 6$ & $\sigma_ 7$ \\
  \hline
F-max   & 1.000 & 0.997 & 0.979 & 0.857 & 0.609 & 0.410 & 0.294 \\ 
  p-min   & 0.000 & 0.000 & 0.000 & 0.000 & 0.000 & 0.000 & 0.000 \\ 
  ERL   & 0.980 & 0.970 & 0.977 & 0.967 & 0.930 & 0.844 & 0.751 \\ 
  Cont   & 0.980 & 0.968 & 0.960 & 0.865 & 0.662 & 0.491 & 0.335 \\ 
  Area   & 0.980 & 0.970 & 0.975 & 0.947 & 0.867 & 0.714 & 0.577 \\ 
   \hline
\end{tabular}
\caption{\label{M1d}Empirical powers of all studied methods and 7 standard deviations for model M1 with error (e).}
\end{table}

\begin{table}[ht]
\centering
\begin{tabular}{rrrrrrrr}
  \hline
 & $\sigma_ 1$ & $\sigma_ 2$ & $\sigma_ 3$ & $\sigma_ 4$ & $\sigma_ 5$ & $\sigma_ 6$ & $\sigma_ 7$ \\  \hline
F-max   & 1.000 & 1.000 & 1.000 & 0.998 & 0.632 & 0.296 & 0.143 \\ 
  $p$-min   & 0.000 & 0.000 & 0.000 & 0.000 & 0.000 & 0.000 & 0.000 \\ 
%  FDR   & 1.000 & 1.000 & 1.000 & 0.701 & 0.001 & 0.001 & 0.000 \\ 
%  $p$-min-rand   & 0.102 & 0.096 & 0.105 & 0.084 & 0.090 & 0.095 & 0.075 \\ 
  ERL   & 0.981 & 0.980 & 0.985 & 0.985 & 0.776 & 0.325 & 0.126 \\ 
  Cont   & 0.981 & 0.980 & 0.985 & 0.971 & 0.522 & 0.249 & 0.129 \\ 
  Area   & 0.981 & 0.980 & 0.985 & 0.984 & 0.749 & 0.340 & 0.141 \\ 
   \hline
\end{tabular}
\caption{\label{M1e}Empirical powers of all studied methods and 7 standard deviations for model M1 with error (f).}
\end{table}

\begin{table}[ht]
\centering
\begin{tabular}{rrrrrrrr}
  \hline
 & $\sigma_ 1$ & $\sigma_ 2$ & $\sigma_ 3$ & $\sigma_ 4$ & $\sigma_ 5$ & $\sigma_ 6$ & $\sigma_ 7$ \\    \hline
F-max   & 0.999 & 0.950 & 0.841 & 0.658 & 0.522 & 0.421 & 0.352 \\ 
  $p$-min   & 0.000 & 0.000 & 0.000 & 0.000 & 0.000 & 0.000 & 0.000 \\ 
%  FDR   & 1.000 & 1.000 & 0.999 & 0.988 & 0.852 & 0.523 & 0.254 \\ 
%  $p$-min-rand   & 0.094 & 0.084 & 0.089 & 0.088 & 0.093 & 0.084 & 0.081 \\ 
  ERL   & 0.986 & 0.980 & 0.977 & 0.979 & 0.981 & 0.984 & 0.978 \\ 
  Cont   & 0.985 & 0.917 & 0.790 & 0.604 & 0.448 & 0.358 & 0.300 \\ 
  Area   & 0.986 & 0.980 & 0.977 & 0.979 & 0.981 & 0.979 & 0.970 \\ 
   \hline
\end{tabular}
\caption{\label{M1f}Empirical powers of all studied methods and 7 standard deviations for model M1 with error (g).}
\end{table}

\begin{table}[ht]
\centering
\begin{tabular}{rrrrrrrr}
  \hline
 & $\sigma_ 1$ & $\sigma_ 2$ & $\sigma_ 3$ & $\sigma_ 4$ & $\sigma_ 5$ & $\sigma_ 6$ & $\sigma_ 7$ \\
  \hline
F-max   & 1.000 & 1.000 & 0.917 & 0.277 & 0.089 & 0.068 & 0.045 \\ 
  $p$-min   & 0.000 & 0.000 & 0.000 & 0.000 & 0.000 & 0.000 & 0.000 \\ 
%  FDR   & 0.017 & 0.001 & 0.000 & 0.000 & 0.000 & 0.000 & 0.000 \\ 
%  $p$-min-rand   & 0.070 & 0.089 & 0.085 & 0.071 & 0.053 & 0.054 & 0.052 \\ 
  ERL   & 0.975 & 0.976 & 0.732 & 0.201 & 0.071 & 0.064 & 0.052 \\ 
  Cont   & 0.975 & 0.975 & 0.799 & 0.230 & 0.077 & 0.058 & 0.048 \\ 
  Area   & 0.975 & 0.980 & 0.825 & 0.254 & 0.085 & 0.068 & 0.049 \\ 
   \hline
\end{tabular}
\caption{\label{M1'a}Empirical powers of all studied methods and 7 standard deviations for model M1' with error (a).}
\end{table}

\begin{table}[ht]
\centering
\begin{tabular}{rrrrrrrr}
  \hline
 & $\sigma_ 1$ & $\sigma_ 2$ & $\sigma_ 3$ & $\sigma_ 4$ & $\sigma_ 5$ & $\sigma_ 6$ & $\sigma_ 7$ \\    \hline
F-max   & 1.000 & 1.000 & 0.906 & 0.280 & 0.107 & 0.076 & 0.060 \\ 
  $p$-min   & 0.000 & 0.000 & 0.000 & 0.000 & 0.000 & 0.000 & 0.000 \\ 
%  FDR   & 0.017 & 0.000 & 0.000 & 0.000 & 0.000 & 0.000 & 0.000 \\ 
%  $p$-min-rand   & 0.076 & 0.074 & 0.077 & 0.079 & 0.047 & 0.055 & 0.041 \\ 
  ERL   & 0.980 & 0.962 & 0.705 & 0.193 & 0.076 & 0.071 & 0.053 \\ 
  Cont   & 0.980 & 0.962 & 0.811 & 0.232 & 0.099 & 0.068 & 0.070 \\ 
  Area   & 0.980 & 0.968 & 0.830 & 0.243 & 0.100 & 0.076 & 0.061 \\ 
   \hline
\end{tabular}
\caption{\label{M1f}Empirical powers of all studied methods and 7 standard deviations for model M1' with error (b).}
\end{table}

\begin{table}[ht]
\centering
\begin{tabular}{rrrrrrrr}
  \hline
 & $\sigma_ 1$ & $\sigma_ 2$ & $\sigma_ 3$ & $\sigma_ 4$ & $\sigma_ 5$ & $\sigma_ 6$ & $\sigma_ 7$ \\
  \hline
F-max   & 1.000 & 1.000 & 1.000 & 1.000 & 1.000 & 0.965 & 0.377 \\ 
  p-min   & 0.000 & 0.000 & 0.000 & 0.000 & 0.000 & 0.000 & 0.000 \\ 
  ERL   & 0.982 & 0.976 & 0.984 & 0.975 & 0.978 & 0.978 & 0.972 \\ 
  Cont   & 0.982 & 0.976 & 0.984 & 0.975 & 0.978 & 0.686 & 0.339 \\ 
  Area   & 0.982 & 0.976 & 0.984 & 0.975 & 0.978 & 0.978 & 0.915 \\ 
 \hline
\end{tabular}
\caption{\label{M1'c}Empirical powers of all studied methods and 7 standard deviations for model M1' with error (c).}
\end{table}

\begin{table}[ht]
\centering
\begin{tabular}{rrrrrrrr}
  \hline
 & $\sigma_ 1$ & $\sigma_ 2$ & $\sigma_ 3$ & $\sigma_ 4$ & $\sigma_ 5$ & $\sigma_ 6$ & $\sigma_ 7$ \\
  \hline
F-max   & 1.000 & 0.993 & 0.549 & 0.060 & 0.066 & 0.062 & 0.048 \\ 
  p-min   & 0.000 & 0.000 & 0.000 & 0.000 & 0.000 & 0.000 & 0.000 \\ 
  ERL   & 0.979 & 0.973 & 0.752 & 0.188 & 0.097 & 0.074 & 0.056 \\ 
  Cont   & 0.979 & 0.907 & 0.445 & 0.080 & 0.053 & 0.054 & 0.047 \\ 
  Area   & 0.979 & 0.974 & 0.800 & 0.194 & 0.085 & 0.072 & 0.048 \\ 
   \hline
\end{tabular}
\caption{\label{M1'd}Empirical powers of all studied methods and 7 standard deviations for model M1' with error (d).}
\end{table}

\begin{table}[ht]
\centering
\begin{tabular}{rrrrrrrr}
  \hline
 & $\sigma_ 1$ & $\sigma_ 2$ & $\sigma_ 3$ & $\sigma_ 4$ & $\sigma_ 5$ & $\sigma_ 6$ & $\sigma_ 7$ \\
  \hline
F-max   & 0.782 & 0.496 & 0.351 & 0.190 & 0.113 & 0.086 & 0.065 \\ 
  p-min   & 0.000 & 0.000 & 0.000 & 0.000 & 0.000 & 0.000 & 0.000 \\ 
  ERL   & 0.752 & 0.505 & 0.384 & 0.235 & 0.167 & 0.084 & 0.093 \\ 
  Cont   & 0.791 & 0.526 & 0.371 & 0.221 & 0.134 & 0.089 & 0.072 \\ 
  Area   & 0.826 & 0.583 & 0.422 & 0.259 & 0.163 & 0.093 & 0.087 \\ 
   \hline
\end{tabular}
\caption{\label{M1'd}Empirical powers of all studied methods and 7 standard deviations for model M1' with error (e).}
\end{table}

\begin{table}[ht]
\centering
\begin{tabular}{rrrrrrrr}
  \hline
 & $\sigma_ 1$ & $\sigma_ 2$ & $\sigma_ 3$ & $\sigma_ 4$ & $\sigma_ 5$ & $\sigma_ 6$ & $\sigma_ 7$ \\
  \hline
F-max   & 1.000 & 1.000 & 0.963 & 0.358 & 0.089 & 0.058 & 0.057 \\ 
  $p$-min   & 0.000 & 0.000 & 0.000 & 0.000 & 0.000 & 0.000 & 0.000 \\ 
%  FDR   & 0.017 & 0.003 & 0.002 & 0.000 & 0.002 & 0.001 & 0.001 \\ 
%  $p$-min-rand   & 0.091 & 0.079 & 0.113 & 0.077 & 0.062 & 0.045 & 0.053 \\ 
  ERL   & 0.975 & 0.978 & 0.720 & 0.110 & 0.059 & 0.051 & 0.045 \\ 
  Cont   & 0.975 & 0.979 & 0.899 & 0.275 & 0.083 & 0.062 & 0.047 \\ 
  Area   & 0.975 & 0.979 & 0.924 & 0.258 & 0.077 & 0.056 & 0.048 \\ 
   \hline
\end{tabular}
\caption{\label{M1'e}Empirical powers of all studied methods and 7 standard deviations for model M1' with error (f).}
\end{table}

\begin{table}[ht]
\centering
\begin{tabular}{rrrrrrrr}
  \hline
 & $\sigma_ 1$ & $\sigma_ 2$ & $\sigma_ 3$ & $\sigma_ 4$ & $\sigma_ 5$ & $\sigma_ 6$ & $\sigma_ 7$ \\
   \hline
F-max   & 0.412 & 0.228 & 0.146 & 0.101 & 0.084 & 0.073 & 0.067 \\ 
  $p$-min   & 0.000 & 0.000 & 0.000 & 0.000 & 0.000 & 0.000 & 0.000 \\ 
%  FDR   & 0.000 & 0.000 & 0.001 & 0.000 & 0.000 & 0.001 & 0.000 \\ 
%  $p$-min-rand   & 0.088 & 0.077 & 0.092 & 0.085 & 0.090 & 0.073 & 0.070 \\ 
  ERL   & 0.894 & 0.589 & 0.311 & 0.193 & 0.186 & 0.132 & 0.120 \\ 
  Cont   & 0.408 & 0.189 & 0.124 & 0.094 & 0.081 & 0.063 & 0.073 \\ 
  Area   & 0.924 & 0.612 & 0.424 & 0.292 & 0.208 & 0.154 & 0.134 \\ 
   \hline
\end{tabular}
\caption{\label{M1'f}Empirical powers of all studied methods and 7 standard deviations for model M1' with error (g).}
\end{table}

\end{document}